\documentclass[12pt]{article}
\usepackage[margin=1.25in]{geometry}

\usepackage{amssymb}
\usepackage{amsmath}
\usepackage{graphicx}
\usepackage{wrapfig}
\usepackage{authblk}
\usepackage{bbm}
\usepackage{float}
\usepackage{tikz}
\usetikzlibrary{arrows,calc,shapes,decorations.pathreplacing}
\usepackage{wrapfig}
\usepackage{hyperref}
\hypersetup{pdftex,colorlinks=true,allcolors=blue}
\usepackage{hypcap}
\usepackage{caption}
\usepackage{pgfplots}
\usepackage{subcaption}

\usepackage{algorithm}
\usepackage{algorithmicx}
\usepackage{algpseudocode}

\newenvironment{customlegend}[1][]{%
    \begingroup
    \csname pgfplots@init@cleared@structures\endcsname
    \pgfplotsset{#1}%
}{%
    \csname pgfplots@createlegend\endcsname
    \endgroup
}%

\def\addlegendimage{\csname pgfplots@addlegendimage\endcsname}

\newtheorem{theorem}{Theorem}[section]
\newtheorem{lemma}[theorem]{Lemma}
\newtheorem{proposition}[theorem]{Proposition}
\newtheorem{corollary}[theorem]{Corollary}

\newtheorem{conjecture}[theorem]{Conjecture}
\newtheorem{assumption}[theorem]{Assumption}

\newtheorem{remark}[theorem]{Remark}
\newtheorem{definition}[theorem]{Definition}

\newenvironment{proof}[1][Proof]{\begin{trivlist}
\item[\hskip \labelsep {\bfseries #1}]}{\end{trivlist}}

\def\P{\mathop{\mathrm{P}}}%{\mbox{P}}
\def\E{\mathrm{E}}
\def\W{\mathrm{W}}
\def\M{\mathrm{M}}
\def\N{\mathrm{N}}

\DeclareMathOperator*{\argmin}{arg\,min}

\newcommand{\qed}{\nobreak \ifvmode \relax \else
      \ifdim\lastskip<1.5em \hskip-\lastskip
      \hskip1.5em plus0em minus0.5em \fi \nobreak
      \vrule height0.75em width0.5em depth0.25em\fi}

\begin{document}

\title{Strategic Bidding in an Accumulating Priority Queue: Equilibrium Analysis\footnote{To appear in the Annals of Operations Research}}

\author{Moshe Haviv and Liron Ravner\footnote{liron.ravner@mail.huji.ac.il}}
\affil{Department of Statistics and the Federmann Center for the Study of Rationality  \\ The Hebrew University of Jerusalem}

\date{\today}
\maketitle

\begin{abstract}
We study the strategic purchasing of priorities in a time-dependent accumulating priority M/G/$1$ queue. We formulate a non-cooperative game in which customers purchase priority coefficients with the goal of reducing waiting costs in exchange. The priority of each customer in the queue is a linear function of the individual waiting time, with the purchased coefficient being the slope. The unique pure Nash equilibrium is solved explicitly for the case with homogeneous customers. A general characterisation of the Nash equilibrium is provided for the heterogeneous case. It is shown that both avoid the crowd and follow the crowd behaviours are prevalent, within class types and between them. We further present a pricing mechanism that ensures the order of the accumulating priority rates in equilibrium follows a $C\mu$ type rule and improves overall efficiency.
\end{abstract}

%%%%%%%%%%%%%%%%%%%%%%%%%%%%%%%%%%%%%%%%%%%%%%%%%%%%%%%%%%%%%%%
\section{Introduction}\label{sec:Intro}
Service systems are often required to serve customers with heterogeneous characteristics, such as arrival rates, service demand, and waiting time sensitivity. A standard practice to address heterogeneity is incorporating priorities into the service regime. The most common priority regime is absolute priority in which customers are assigned priority and are admitted to service only when there is no higher priority customer in the queue. While it can be optimal in the sense of minimizing expected waiting times, when applying the well known $C\mu$ rule, the major shortcoming of this regime is that it can lead to very high expected waiting times for low priority customers. This may be undesirable if the system seeks to achieve some ``fairness'' objective too. Furthermore, in healthcare systems the condition of the patients may deteriorate while they wait. A possible modification of this regime is to allow customers to accumulate priority while waiting in the queue at varying rates, determined by their priority class (see \cite{K1964}). Our work introduces an economic analysis of an Accumulating Priority (AP) Queue, in which customers purchase their linear rates of accumulation.

We consider an unobservable M/G/$1$ queue with a non-preemptive AP regime. There is a discrete number of customer types, who may differ in their arrival rates, service distributions and linear waiting costs. Upon arrival every customer purchases a linear AP coefficient, referred to as a bid, knowing her own type but not the system state. All customers share the belief that the system is in steady state upon their arrival. This results in a non-cooperative game, and our goal is to characterise and analyse its Nash equilibrium. We explicitly compute the unique symmetric equilibrium when all customers have the same waiting costs, but not necessarily the same service distribution. It is further shown that both avoid the crowd (ATC) and follow the crowd (FTC) behaviours are possible for different bidding levels. We then proceed to a characterisation of the symmetric (within type classes) equilibrium in the general case. The equilibrium is pure and given by a solution to a system of non-linear equations that can be represented by two recursive formulas for the expected waiting times that need to be satisfied simultaneously. We show that ATC and FTC between class types are both possible for different bidding levels. The order of the bids in equilibrium is determined by the waiting time costs, and not the service moments. This may potentially yield very far from optimal results in terms of expected waiting times. We therefore suggest a simple service time based pricing mechanism in order to achieve a balance between the AP regime constraint and the socially optimal absolute priority regime.

The accumulating priority regime was introduced for the M/M/$1$ queue by Kleinrock in \cite{K1964}, with a more detailed analysis (for the M/G/$1$ queue) in the context of other priority regimes later appearing in p126 of \cite{book_K1976}. The main result was a recursive formula for the expected waiting times of the different priority classes. Note that it was then referred to as the time-dependent priority regime, but this has been used in other contexts over the years so we use instead the terminology of Stanford et al{.} in \cite{STZ2014}. The latter work presented a rigorous analysis of the waiting time distributions of the different classes via their Laplace transforms. A multi-server extension was studied by Sharif et al{.} in \cite{SSTZ2014} with the additional assumption of exponential service times. Several works provided extensions and generalizations of the mean value analysis for non-linear AP rates: power law \cite{KF1967}, affine \cite{G1977}, concave \cite{NA1979}, and negative linear \cite{K1982}. When departing from the linear model there are no longer any closed form solutions for the expected waiting times.

Analysis of the social optimization problem for absolute priorities can also be found in p135 of Kleinrock \cite{book_K1976}. Balachandran \cite{B1972} showed that ``stable", in a game equilibrium sense, purchasing of absolute priority is not socially optimal. Glazer and Hassin provided explicit analysis of this model in \cite{GH1986}. Haviv and van der Wal in \cite{HV1997} studied the game of purchasing priorities in an M/M/$1$ queue for two related regimes, namely relative priority and random order that are both determined by the customer bids. Mendelson and Whang suggested an incentive compatible pricing mechanism for the purchasing of absolute priorities in \cite{MW1990}. Hassin \cite{H1995} analysed the decentralized regulation of an absolute priority queue in which customers have the option of balking, on top of the purchasing of priorities. In \cite{AM2004}, Af{\`e}che and Mendelson presented a general model with dependent service and delay valuations, and analysed the optimal pricing mechanism for both preemptive and non-preemptive absolute priority regimes. Our work makes a contribution to this body of knowledge in the context of a dynamic priority queueing regime.

%%%%%%%%%%%%%%%%%%%%%%%%%%%%%%%%%%%%%%%%%%%%%%%%%%%%%%%%%%%%%%%
\section{Queueing model and preliminaries}\label{sec:model}
Suppose $N\in\mathbbm{N}_+$ types of customers arrive to a single server queue according to independent Poisson processes with rates $\lambda_i$ for $i=1,\ldots,N$, such that $\lambda:=\sum_{i=1}^N \lambda_i$. We assume that the service time of any type $i=1\ldots,N$ customer is a random variable $X_i\sim G_i$ (independent of all other service and arrival times), where $G_i$ is some general \textit{cdf}. Further denote the $k$'th moment of the service time by $\overline{x^k}_i:=\E X_i^k$ for $i=1,\ldots,N$. Customers are assigned an accumulating priority (AP) rate $b_i$ for every unit of waiting time according to their class type, $i=1,\ldots,N$. Without loss of generality we assume $b_1\leq b_2\leq \cdots\leq b_N$. The AP at time $t$ of a type $i$ customer who arrived at time $s$, $s \leq t$,  is $b_i(t-s)$. Upon a service completion, if the queue is not empty then the server admits the customer with the highest accumulated priority, and completes her service without preemption.  An example of the dynamics of the accumulating priority regime is illustrated in Figure \ref{fig:AP_example}. If $b_i=b, \ \forall i=1,\ldots,N$, then clearly the queue is a standard First-Come First-Served (\textit{FCFS}) M/G/$1$.

\begin{figure}[H]
\centering
\begin{tikzpicture}[xscale=2.8,yscale=1]
  \def\xmin{0}
  \def\xmax{2.5}
  \def\ymin{0}
  \def\ymax{2.5}
    \draw[->] (\xmin,\ymin) -- (\xmax,\ymin) node[right] {$t$} ;
    \draw[->] (\xmin,\ymin) -- (\xmin,\ymax) node[above] {AP} ;
    \foreach \x in {0,1,2}
    \node at (\x,\ymin) [below] {\x};
    \foreach \y in {1,2}
    \node at (\xmin,\y) [left] {\y};

    \draw[->,red, thick,domain=0:2.5]  plot (\x, {0.5*\x});
    \draw[->,blue,dashed, thick,domain=1:2.5]  plot (\x, {\x-1});
    \draw[smooth,black, thick,dotted]  (2,0) -- (2,1);
    \draw[smooth,black, thick,dotted]  (0,1) -- (2,1);

	\draw[] (2.4,0.95) node[right,red] {$b_1 t$};
	\draw[] (2.4,1.75) node[right,blue] {$b_2 t$};
	
\end{tikzpicture}
\caption{Example of the priority evolution for a type $1$ customer arriving at $t=0$ (solid red), and a type $2$ customer arriving at $t=1$ (dashed blue), with coefficients $b_1=0.5$ and $b_2=1$, respectively. At any time $t<2$ the type $1$ customer has higher priority than the type $2$ customer, but from $t=2$ and on she has lower priority and will effectively be overtaken.}\label{fig:AP_example}
\end{figure}
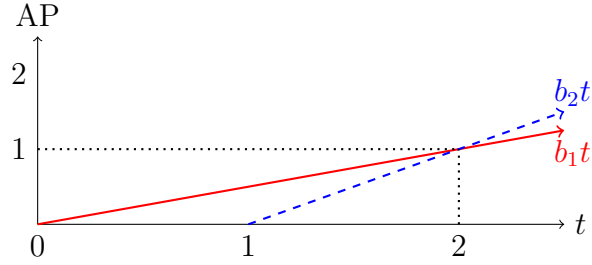

Let $W_i$ be the stationary waiting time of a type $i=1\ldots,N$ customer, excluding service, and denote $\W_i:=\E W_i$. We assume the stability condition $\rho:=\sum_{i=1}^N \rho_i<1$, where $\rho_i:=\lambda_i\overline{x}_i$. Hence, from p131 of Kleinrock \cite{book_K1976}, the expected waiting times satisfy the recursive formula,
\begin{equation}\label{eq:model_waiting_i}
\W_i=\frac{\frac{\W_0}{(1-\rho)}-\sum_{k=1}^{i-1}\rho_k(1-\frac{b_k}{b_i})\W_k}{1-\sum_{k=i}^N \rho_k(1-\frac{b_i}{b_k})}, \ 1 \leq i \leq N,
\end{equation}
where $\W_0=\sum_{i=1}^N\frac{\rho_i\overline{x^2}_i}{2\overline{x}_i}$ is the expected residual service of the customer in service. This formula allows for an easy computation of the expected waiting times given the AP rates $b_1,\ldots,b_N$.

In order to introduce economic analysis we assume that customers are sensitive to waiting in the queue. Specifically, customers of type $i=1,\ldots,N$ incur a cost of $C_i>0$ per unit of time in the queue (not including service). The cost parameter may be homogeneous for all customers or may depend on their type, and we will analyse both scenarios.

For the game analysis in the following sections it will be useful to consider a straightforward generalization of \eqref{eq:model_waiting_i}. Specifically, the case where every arriving type $i=1,\ldots,N$ customer is independently assigned a priority parameter according to a r.v. $B_i\sim F_i$. To avoid technicalities we make a standard assumption on the priority  rate distribution. We then proceed to analyse the general properties of the expected waiting times, namely continuity and convexity with respect to a single AP coefficient. We denote the expected waiting time of a customer with realised AP rate $a$, given the profile of distributions $\mathcal{F}=\{F_i, \ i=1,\ldots,N\}$, by $\W(a;\mathcal{F})$. Take note that this expectation is not dependent on the type of the customer.

\begin{assumption}\label{assumption:F}
The priorities are strictly positive, $F(0)=0$, and $F$ is comprised of a countable collection of atoms and continuous intervals with a positive density.
\end{assumption}

\begin{lemma}\label{lemma:W_general}
Let $\mathcal{F}=\{F_i, \ i=1,\ldots,N\}$ such that $F_i$ satisfies Assumption \ref{assumption:F} for all $i=1,\ldots,N$.
\begin{enumerate}
\item The expected waiting time of a customer with an AP rate of a satisfies:
\begin{equation}\label{eq:model_waiting_F}
\W(a;\mathcal{F})=\frac{\frac{\W_0}{1-\rho}-\sum_{i=1}^N\int_0^{a} \rho_i\W(b;\mathcal{F})\left(1-\frac{b}{a}\right)\ dF_i(b)}{1-\sum_{i=1}^N\int_a^\infty \rho_i\left(1-\frac{a}{b}\right)\ dF_i(b)}.
\end{equation}
\item $\W(a;\mathcal{F})$ and $\frac{d}{da}\W(a;\mathcal{F})$ are continuous with respect to $a$.
\item $\W(a;\mathcal{F})$ is decreasing and strictly convex with respect to $a$.
\end{enumerate}
\end{lemma}
We leave the proof of this generalization for the appendix. Observe that the original discrete formulation of \eqref{eq:model_waiting_i} is a special case, where $F_i(a)=\mathrm{1}_{\{a\geq b_i\}}, \ \forall i=1,\ldots,N$. The importance of Lemma \ref{lemma:W_general} is that if customers have linear waiting costs, then the cost function they wish to minimize (given~$\mathcal{F}$) is strictly convex with well defined first-order conditions with respect to their individual AP rate, denoted here by $a$. Note that the lemma and the above conclusion also hold for a continuum of customer types with any joint distribution of service and AP rates.

%%%%%%%%%%%%%%%%%%%%%%%%%%%%%%%%%%%%%%%%%%%%%%%%%%%%%%%%%%%%%%%
\section{Purchasing priorities}\label{sec:pp}
The question now is what if customers have to pay for their AP rates instead of having them exogenously assigned? To answer this question we formulate a game in which customers with linear waiting sensitivity can purchase priorities. We assume that there are $N$ customer types with linear waiting costs per unit of time, $C_1<C_2<\cdots<C_N$. An individual customer's action is a bid $b>0$ for priority rate given her type. Apart from the knowledge of one's type we assume all customers share the distributional belief of the process being in steady state. A mixed strategy is a probability distribution on the non-negative real values. A symmetric (within classes) strategy profile is a set of distributions, $\mathcal{F}=\{F_i,\  i=1,\ldots,N\}$, specifying a single bidding strategy for every type of customer. The cost incurred by a type $i$ customer paying $b$ given that all other customers pay according to $\mathcal{F}$ is
\begin{equation}\label{pp:cost}
c_i(b;\mathcal{F}):=C_i\W(b;\mathcal{F})+b, \quad i=1,\ldots,N.
\end{equation}

A symmetric strategy profile $\mathcal{F}$ is a Nash equilibrium if $F_i$ is a best response to $\mathcal{F}$ for any customer of type $i=1,\ldots,N$. We will next assert that an equilibrium exists, and further that any equilibrium is pure. The existence is verified using the conditions of Kakutani's Fixed-Point Theorem. The fact that the equilibrium is in pure strategies is due to the the strict convexity of the waiting cost with respect to a customer's bid.

According to Lemma \ref{lemma:W_general}, the waiting time of any customer is decreasing and strictly convex with her own AP rate. This means that the cost function given by \eqref{pp:cost} is strictly convex, and therefore has a unique minimizer $b>0$ for any strategy profile $\mathcal{F}$ of the other customers. In particular, any equilibrium point is pure as by definition as all customers are playing a cost minimizing strategy. Thus, from now on we focus on pure strategy profiles given by $\mathbf{b}=(b_1,\ldots,b_N)\in\mathbbm{R}^N$, where $b_i$ is a deterministic action for $i=1,\ldots,N$. Let $r_i(\mathbf{b})$ be the best response of a customer of type $i$ to strategy profile $\mathbf{b}$.

\begin{definition}\label{def:NE}
A pure strategy profile $\mathbf{b}^e$ is a symmetric Nash equilibrium if
\[
b_i^e=r_i(\mathbf{b}^e),\quad \forall \ i=1,\ldots,N.
\]
\end{definition}

The game is non-atomic in the sense that the action of a single customer has no impact on the steady-state properties of the system. This implies that all individual customers of the same type face the same cost function, and therefore also have the same unique best response. Hence, any equilibrium is symmetric (within type classes) and we do not lose generality by considering only symmetric strategy profiles. Also, note that no customer will bid higher than
\[
\overline{b}_i:=C_i\frac{\lambda\overline{x^2}+2(1-\rho)\overline{x}}{2(1-\rho)^2}, \ i=1,\ldots,N,
\]
which is the waiting cost incurred by customer with cost $C_i$ with absolute zero priority (see \cite{book_H2013} pp. 60). We can therefore limit the actions of the customers to the compact and convex action space given by the Cartesian product,
\[
\mathcal{B}:=\prod_{i=1}^N[0,\overline{b}_i]\subset \mathbbm{R}^N.
\]

As we have already stated, the best response function $\mathbf{r}(\mathbf{b}):=\{r_i(\mathbf{b}),\ i=1,\ldots,N\}$ returns a unique value for each coordinate, hence $\mathbf{r}(\mathbf{b})\in\mathbbm{R}^N$ is trivially a convex set for any $\mathbf{b}$. Finally, $\mathbf{r}$ also has a closed graph. This can be argued in two steps: first of all, the convexity of the cost function (part 3 of Lemma \ref{lemma:W_general}) implies that the best response function is given by the first-order conditions, secondly the first-order conditions are continuous (part 2 of Lemma \ref{lemma:W_general}). To summarize:
\begin{itemize}
\item $\mathcal{B}\subset\mathbbm{R}^N$ is a compact and closed set.
\item $\mathbf{r}(\mathbf{b})$ is a convex set for any $\mathbf{b}\in\mathcal{B}$.
\item $\mathbf{r}:\mathcal{B}\to\mathcal{B}$ has a closed graph.
\end{itemize}
 
We have thus established that the existence conditions of Kakutani's Fixed-Point Theorem (see \cite{book_OR1994}), yielding the following general existence Lemma for our model.

\begin{lemma}\label{lemma:pp_existence}
There exists a (within class) symmetric pure strategy equilibrium.
\end{lemma}

%%%%%%%%%%%%%%%%%%%%%%%%%%%%%%%%%%%%%%%%%%%%%%%%%%%%%%%%%%%%%%%
\section{Homogeneous customers}\label{sec:pp_homogeneous}
We now assume that there is only a single customer type, i.e. $N=1$, the total arrival rate is $\lambda$, the service distribution is $G$, and all customers incur waiting cost $C$ per unit of time in the queue.
\begin{theorem}\label{thm:pp_homo_equil}
If $N=1$, then the unique pure Nash equilibrium bid is
\begin{equation}\label{eq:pp_homo_equil}
b^e=\frac{C\rho\W_0}{1-\rho}.
\end{equation}
In particular, the resulting queueing process is that of a \textit{FCFS} regime where all customers incur a cost of $C(1+\rho)\frac{\W_0}{1-\rho}$.
\end{theorem}
\begin{proof}
We commence by analysing the best response of a single customer when all other customers make the same bid $b$. If the customer bids $a$ then we can use \eqref{eq:model_waiting_i} for $N=2$, with arrival rates zero (for the tagged customer) and $\lambda$ (for all others), to obtain her expected waiting time,
\[
\W(a;b)=
\left\{
	\begin{array}{ll}
		\frac{\mathrm{W}_0}{1-\rho}\frac{b}{(1-\rho)b+\rho a} \mbox{, } &  a<b, \\
		\frac{\mathrm{W}_0}{1-\rho} \mbox{, } &  a=b,  \\
		\frac{\mathrm{W}_0}{1-\rho}\frac{(1-\rho)a+\rho b}{a} \mbox{, } &  a>b.
	\end{array}
\right.
\]

The derivative of the waiting time with respect to $a$ for $a\neq b$ is then
\begin{equation}\label{eq:pp_Deriv}
\frac{d}{da}\W(a;b)=
\left\{
	\begin{array}{ll}
		-\frac{\mathrm{W}_0}{1-\rho}\frac{\rho b}{((1-\rho)b+\rho a)^2} \mbox{, } &  a<b, \\
		-\frac{\mathrm{W}_0}{1-\rho}\frac{\rho b}{a^2} \mbox{, } &  a>b.
	\end{array}
\right.
\end{equation}

It is easy to verify directly that the second derivative is positive in both cases, but this is not necessary as convexity was generally established in Lemma \ref{lemma:W_general}. Furthermore, it was shown that the derivative is continuous, hence by~\eqref{eq:pp_Deriv} we have that
\begin{equation}\label{eq:pp_Deriv_fixed}
\lim_{a\uparrow b}\frac{d}{da}\W(a;b)=\lim_{a\downarrow b}\frac{d}{da}\W(a;b)=-\frac{\rho\mathrm{W}_0}{(1-\rho)b}.
\end{equation}

The cost minimizing bid is given by the first-order condition, $\frac{d}{da}c(a;b)=0$, where $c(a;b):=C\W(a;b)+a$. The best response, $r(b)=\{a:\frac{d}{da}c(a;b)=0\}$, can be obtained for $a<b$ and $a>b$ by solving the respective constrained quadratic equations, however it is not necessary for finding the equilibrium value of $b$. We are interested in the fixed point $b^e=r(b^e)$, and from \eqref{eq:pp_Deriv_fixed} we get the equilibrium bid in \eqref{eq:pp_homo_equil}. Therefore, all customers bidding $b^e$ and incurring a cost of $C(1+\rho)\frac{\W_0}{1-\rho}$ is the unique equilibrium of the game.
\qed
\end{proof}

\begin{remark}
The equilibrium remains unchanged even if customers are heterogeneous in their service time distributions. That is, the general model with $\lambda_i$ and $G_i$ for $i=1,\ldots,N$, and homogeneous costs $C_i=C$. This is due to the non-preemptive regime which implies that the waiting time in the queue is not affected by one's own service time. 
\end{remark}

We next proceed to highlight an interesting feature of this game, namely the non-monotone behaviour of the best-response function, $r(b):=\argmin_{a>0}c(a;b)$. The game is said to satisfy the avoid the crowd (ATC) property or follow the crowd (FTC) property if $r(b)$ is monotone decreasing or increasing, respectively (see p6 of \cite{book_HH2003}). We show that our game is neither ATC or FTC, but rather displays monotone increasing behaviour for $b\leq\tilde{b}$ and monotone decreasing behaviour for $b>\tilde{b}$ where $\tilde{b}:=b^e\max\left\lbrace 1,\frac{1}{4(1-\rho)^2}\right\rbrace$. We formally characterise the best response function and ATC/FTC behaviour in the following proposition. We assume that $b>0$ and later comment on the case $b=0$.

\begin{proposition}\label{prop:pp_BR}
If all customers bid $b>0$ then
\begin{equation}\label{eq:pp_BR}
r(b)=
\left\{
	\begin{array}{ll}
		\sqrt{b^e b} \mbox{, } &  0<b<b^e, \\
		\frac{\sqrt{b^e b}-(1-\rho)b}{\rho} \mbox{, } &  b^e\leq b<\frac{b^e}{(1-\rho)^2}, \\
		0 \mbox{, } & b\geq \frac{b^e}{(1-\rho)^2} .
	\end{array}
\right.
\end{equation}
Moreover, if $\rho\leq \frac{1}{2}$ then
\begin{itemize}
\item $r(b)$ is monotone increasing (FTC)
for $0 < b < b^e$,
\item $r(b)$ is monotone decreasing (ATC) for $b^e < b < b^e/(1-\rho)^2$,
\end{itemize}

and if $\rho>\frac{1}{2}$, then
\begin{itemize}
\item $r(b)$ is monotone increasing (FTC)
for $0 < b < b^e/(4(1-\rho)^2)$,
\item $r(b)$ is monotone decreasing (ATC) for $b^e/(4(1-\rho)^2) < b < b^e/(1-\rho)^2$.
\end{itemize}
Observe that $b^e<b^e/(4(1-\rho)^2)$. Finally, $r(b)=0$ for $b>b^e/(1-\rho)^2$.
\end{proposition}
\begin{proof}
If all customers make bid $b>0$ then from Lemma \ref{lemma:W_general} we have that there is a unique solution to the first-order condition $\frac{d}{da}c(a;b)=0$. In \eqref{eq:pp_Deriv} we saw that the condition is different for $a<b$ and $a>b$:
\[
\begin{array}{ccc}
1 = \frac{b^eb}{((1-\rho)b+\rho a)^2} & \mbox{,} & a<b, \\
1 = \frac{b^eb}{a^2} & \mbox{,} & a>b.
\end{array}
\]
If $b<b^e$ then the second equation has a solution $a$ such that $b<a<b^e$. By the convexity property we have that there is only one solution, and therefore if there is a solution to the second equation then the first equation does not have a solution, and vice versa. If $b>b^e$ then there is no solution to the second equation and the best response is either the solution to the first equation or zero. Simple algebra yields \eqref{eq:pp_BR}. By taking derivative of $r(b)$ w.r.t. $b$ and comparing to zero we get the remainder of the proposition.
\qed
\end{proof}

Figure \ref{fig:pp_BR}  illustrates the shape of the best response function for two numerical examples. In the left-hand side example we set $\rho=\frac{1}{2}$ and the FTC changes to ATC at $b^e$ exactly. In the right-hand side example we set $\rho>\frac{1}{2}$ and the change occurs after $b^e$. Queueing games with FTC behaviour often have multiple pure equilibria, but we have shown that for our model the equilibrium is nevertheless unique. The fact that at some point the best response function becomes ATC is one reason for this. An additional feature of this model is that there is a discontinuity at zero, which prevents the strategy of all bidding zero from being an equilibrium point. If all customers bid zero then a singled out user will always want to bid some, as small as possible, $\epsilon>0$ and get absolute priority. Note that all bidding zero is socially optimal since the customers are homogeneous. This implies that our model displays the common phenomenon of rent-dissipation (see p85 of \cite{book_HH2003}).

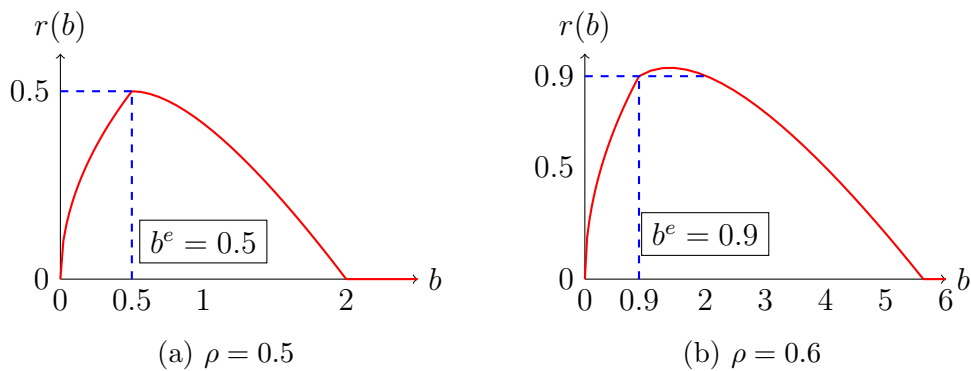
\begin{figure}[H]
\centering
\begin{subfigure}{.45\linewidth}
\centering
\begin{tikzpicture}[xscale=1.9,yscale=5]
  \def\xmin{0}
  \def\xmax{2.5}
  \def\ymin{0}
  \def\ymax{0.6}
    \draw[->] (\xmin,\ymin) -- (\xmax,\ymin) node[right] {$b$} ;
    \draw[->] (\xmin,\ymin) -- (\xmin,\ymax) node[above] {$r(b)$} ;
    \foreach \x in {0,0.5,1,2}
    \node at (\x,\ymin) [below] {\x};
    \foreach \y in {0,0.5}
    \node at (\xmin,\y) [left] {\y};
    \draw[red, thick, domain=0:0.5] plot (\x, {sqrt(0.5*\x)});
    \draw[red, thick, domain=0.5:2] plot (\x, {(sqrt(0.5*\x)-0.5*\x)/0.5});
    \draw[red, thick, domain=2:\xmax] plot (\x, {0});
    \draw[blue, thick, dashed] (0.5,\ymin)--(0.5,0.5);
    \draw[blue, thick, dashed] (\xmin,0.5)--(.5,0.5);
    \node[draw] at (1,0.1) {$b^e=0.5$};
\end{tikzpicture}
\caption{$\rho=0.5$}\label{fig:pp_BR_a}
\end{subfigure}
\begin{subfigure}{.45\linewidth}
\centering
\begin{tikzpicture}[xscale=0.8,yscale=3]
  \def\xmin{0}
  \def\xmax{6}
  \def\ymin{0}
  \def\ymax{1}
    \draw[->] (\xmin,\ymin) -- (\xmax,\ymin) node[right] {$b$} ;
    \draw[->] (\xmin,\ymin) -- (\xmin,\ymax) node[above] {$r(b)$} ;
    \foreach \x in {0,0.9,2,3,4,5,6}
    \node at (\x,\ymin) [below] {\x};
    \foreach \y in {0,0.5,0.9}
    \node at (\xmin,\y) [left] {\y};
    \draw[red, thick, domain=0:0.9] plot (\x, {sqrt(0.9*\x)});
    \draw[red, thick, domain=0.9:5.625] plot (\x, {(sqrt(0.9*\x)-0.4*\x)/0.6});
    \draw[red, thick, domain=5.625:\xmax] plot (\x, {0});
    \draw[blue, thick, dashed] (0.9,\ymin)--(0.9,0.9);
    \draw[blue, thick, dashed] (\xmin,0.9)--(2.1,0.9);
    \node[draw] at (2,0.2) {$b^e=0.9$};
\end{tikzpicture}
\caption{$\rho=0.6$}\label{fig:pp_BR_b}
\end{subfigure}
\caption{Best response functions for parameters $C=1$, $\overline{x}=1$, $\overline{x^2}=2$, and two different levels of $\lambda$.}
\label{fig:pp_BR}
\end{figure}

It is interesting to compare the equilibrium result we have obtained to those of Haviv and van der Wal \cite{HV1997}. They also assumed that customers incur a cost of $C$ per unit of waiting time. Two service regimes were examined: processor sharing of capacity proportional to the purchased priorities and random admittance with probabilities also proportional to the purchased priorities. In both cases they showed that the Nash equilibrium is a pure strategy where all customers purchase the same priority, with bids $\frac{C\W_0}{\mu(1-\rho)(2-\rho)}$ and $\frac{C\rho \W_0}{\mu(1-\rho)(2-\rho)}$, respectively. Note that effectively, i.e. in equilibrium, no customer has any priority over the others in both cases . In particular, the first model results in the \textit{FCFS} regime while the second in the Egalitarian Processer Sharing (EPS) regime. We can conclude that the equilibrium bid under the AP regime, $\frac{C\rho \W_0}{\mu(1-\rho)}$, is higher than that of the relative priority regime, but may  be either higher or lower than under the PS regime, depending on the value of $\rho$.

%%%%%%%%%%%%%%%%%%%%%%%%%%%%%%%%%%%%%%%%%%%%%%%%%%%%%%%%%%%%%%%
\section{Heterogeneous customers}\label{sec:pp_heterogeneous}
We now return to the general model in which there are $N$ types of customers with arrival rates $\lambda_i$ and service distributions $G_i$, for every $i=1,\ldots,N$. As before, we denote the expected waiting time of a customer bidding $a$ while all others bid according to $\mathbf{b}=(b_1,\ldots,b_N)$ (i.e., each class has a respective bid) by $\W(a;\mathbf{b})$. A customer of type $i$ has a waiting cost of $C_i$, and wishes to minimize $C_i\W(a;\mathbf{b})+a$, where $\mathbf{b}$ is the profile being used by all other customers. We seek the Nash equilibrium strategy profile $\mathbf{b}^e:=(b_1^e,\ldots,b_N^e)$.

Recall that according to Lemma \ref{lemma:W_general} there is a unique best response for any customer to the strategies of other customers. Furthermore, this best response is identical for all non-atomic customers of the same type. This implies that any equilibrium profile is pure and symmetric within the classes of types, as was discussed in Section \ref{sec:pp}. The expected waiting time of a customer bidding $a$ when all other users bid according to profile $\mathbf{b}$, regardless of her type, is given by \eqref{eq:model_waiting_i}:
\begin{equation}\label{eq:pp_hetero_waiting}
\W(a,\mathbf{b})=\frac{\frac{\W_0}{(1-\rho)}-\sum_{\{k:b_k<a\}}\rho_k(1-\frac{b_k}{a})\W(b_k,\mathbf{b})}{1-\sum_{\{k:b_k\geq a\}} \rho_k(1-\frac{a}{b_k})},
\end{equation}

\begin{lemma}\label{lemma:b_order}
If $C_1<C_2<\cdots<C_N$, then the equilibrium bids are ordered:
\[
b_1^e<b_2^e<\cdots<b_N^e.
\]
\end{lemma}
\begin{proof}
Assume that $\mathbf{b}$ is an equilibrium strategy profile such that $b_i>b_j$ for some $i<j$. From Lemma \ref{lemma:W_general} we know that the best response bid to any strategy profile is unique. In equilibrium, for any single customer of type $i$ the cost of increasing her bid from $b_j$ to $b_i$ is smaller than the corresponding reduction in waiting cost, i.e.,
\[
b_i-b_j<C_i\left(\W(b_j;\mathbf{b})-\W(b_i;\mathbf{b})\right).
\]
But since $C_i<C_j$, for any single customer of type $j$ we also have that
\[
b_i-b_j<C_j\left(\W(b_j;\mathbf{b})-\W(b_i;\mathbf{b})\right),
\]
which contradicts the assumption that $b_j$ is a best response for a type~$j$ customer. \qed
\end{proof}

The implications of Lemma \ref{lemma:b_order} are not merely technical, as they highlight the fact that the equilibrium order of priorities is determined
solely by the order of the costs ($C_i,\ i=1,\ldots,N$) and not by the expected service times ($\overline{x}_i,\ i=1,\ldots,N$). This can come at a considerable cost in terms of social welfare when classes with higher waiting costs also have higher expected service times. We will elaborate on this issue in the sequel.

Suppose that all customers pay according to an ordered bidding profile $\mathbf{b}$. The necessary conditions of the equilibrium are obtained as follows:
\begin{enumerate}
\item The unique best response, $a$, to $\mathbf{b}$ of a singled out customer of type $i=1,\ldots,N$ is given by the first-order condition, $C_i\frac{d}{da}\W(a;\mathbf{b})+1=0$.
\item The symmetric best response, $a_i$, to $\mathbf{b}_{-i}$ for all customers of type $i$ is given by $C_i\frac{d}{d a}\W(a_i;\mathbf{b}_{-i}\cup a_i)+1=0$.
\item The Nash equilibrium is given by the simultaneous symmetric best response of all types $i=1,\ldots,N$.
\end{enumerate}

The solution to the first-order condition of a singled out user of type $i$ lies in exactly one of the possible $N+1$ intervals defining \eqref{eq:pp_hetero_waiting}. Suppose that the solution is one such that $a\in(b_{i-1},b_{i})$, which is sufficient for equilibrium analysis in light of Lemma \ref{lemma:b_order}. By taking derivative of \eqref{eq:pp_hetero_waiting} we have that for $a\in(b_{i-1},b_{i})$,  $\frac{d}{da}\W(a;\mathbf{b})$ equals
\[
-\frac{\sum_{k=1}^{i-1}\rho_k\frac{b_k}{a^2}\left(1-\sum_{k=i}^N \rho_k\left(1-\frac{a}{b_k}\right)\right)+\sum_{k=i}^N\frac{\rho_k}{b_k}\left(\frac{\W_0}{1-\rho}-\sum_{k=1}^{i-1}\rho_k\left(1-\frac{b_k}{a}\W(b_k;\mathbf{b})\right)\right)}{\left(1-\sum_{k=i}^N \rho_k\left(1-\frac{a}{b_k}\right)\right)^2}.
\]
After applying some algebra, the first-order condition for type $i$ is
\[
1-\sum_{k=i}^N\rho_k\left(1-\frac{a}{b_k}\right)=C_i\left(\frac{1}{a^2}\sum_{k=1}^{i-1}\rho_k b_k \W(b_k;\mathbf{b})+\W(a;\mathbf{b})\sum_{k=i}^N\frac{\rho_k}{b_k}\right).
\]

Plugging in $a=b_i$ we derive that the symmetric first-order condition for all type $i$ customers is $\W(b_i;\mathbf{b})=\tilde{\W}(b_i;\mathbf{b})$, where
\begin{equation}\label{eq:pp_hetero_waiting_NE}
\tilde{\W}(b_i;\mathbf{b})=\frac{1-\sum_{k=i}^N\rho_k\left(1-\frac{b_i}{b_k}\right)-\frac{C_i}{b_i^2}\sum_{k=1}^{i-1}\rho_k b_k \W(b_k,\mathbf{b})}{C_i\sum_{k=i}^N\frac{\rho_k}{b_k}}, \quad i=1,\ldots,N.
\end{equation}

The above analysis yields the main result of this section.

\begin{theorem}\label{thm:pp_hetero_equil}
Any symmetric Nash equilibrium $\mathbf{b}^e:=\{b_1^e,\ldots,b_N^e\}$ is ordered in the same order as the waiting costs, and satisfies 
\[
\W(b_i;\mathbf{b})=\tilde{\W}(b_i;\mathbf{b}), \ \forall i=1,\ldots,N.
\]
\end{theorem}

Note that $\W(b_i;\mathbf{b})$ is given by the recursive expected waiting time solution \eqref{eq:model_waiting_i}, and that $\tilde{\W}(b_i;\mathbf{b})$ is also a recursive formula \eqref{eq:pp_hetero_waiting_NE} in terms of the expected waiting times, and thus an equilibrium solution satisfies both recursions simultaneously.

In Figure \ref{fig:equil_bids} we illustrate the equilibrium solution for an example with five customer types. In particular, the expected waiting times and equilibrium bids are calculated. As the level of utilization grows, all waiting times increase, as expected, and customers increase their bids
as well according to their type.

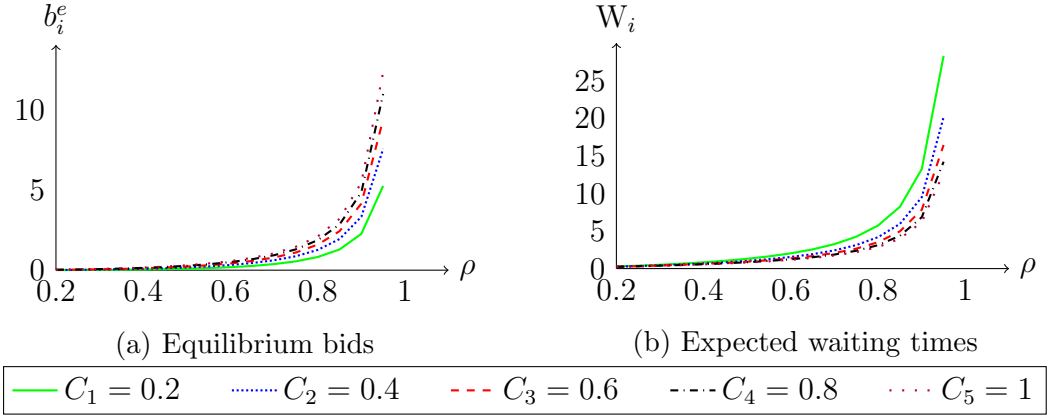
\begin{figure}[H]
\centering
\begin{subfigure}{.48\linewidth}
\centering
\begin{tikzpicture}[xscale=5.8,yscale=0.214]
  \def\xmin{0.2}
  \def\xmax{1.1}
  \def\ymin{0}
  \def\ymax{14}
    \draw[->] (\xmin,\ymin) -- (\xmax,\ymin) node[right] {$\rho$} ;
    \draw[->] (\xmin,\ymin) -- (\xmin,\ymax) node[above] {$b_i^e$} ;
    \foreach \x in {0.2,0.4,0.6,0.8,1}
    \node at (\x,\ymin) [below] {\x};
    \foreach \y in {0,5,10}
    \node at (\xmin,\y) [left] {\y};

    \draw[smooth,green,thick] (0.2,0.005)--	(0.25,0.01)--	(0.3,0.017)--	(0.35,0.027)--	(0.4,0.041)--	(0.45,0.062)--	(0.5,0.091)--	(0.55,0.131)--	(0.6,0.186)--	(0.65,0.265)--	(0.7,0.379)--	(0.75,0.549)--	(0.8,0.819)--	(0.85,1.288)--	(0.9,2.259)--	(0.95,5.238);

    \draw[smooth,blue,densely dotted,thick] (0.2,0.014)--	(0.25,0.025)--	(0.3,0.039)--	(0.35,0.059)--	(0.4,0.086)--	(0.45,0.123)--	(0.5,0.171)--	(0.55,0.236)--	(0.6,0.324)--	(0.65,0.445)--	(0.7,0.616)--	(0.75,0.868)--	(0.8,1.26)--	(0.85,1.936)--	(0.9,3.32)--	(0.95,7.546);

	\draw[smooth,red,dashed,thick] (0.2,0.021)--	(0.25,0.035)--	(0.3,0.055)--	(0.35,0.081)--	(0.4,0.117)--	(0.45,0.165)--	(0.5,0.227)--	(0.55,0.31)--	(0.6,0.421)--	(0.65,0.574)--	(0.7,0.788)--	(0.75,1.101)--	(0.8,1.588)--	(0.85,2.425)--	(0.9,4.14)--	(0.95,9.372);

	\draw[smooth,black,dashdotted,thick] (0.2,0.026)--	(0.25,0.043)--	(0.3,0.067)--	(0.35,0.1)--	(0.4,0.143)--	(0.45,0.199)--	(0.5,0.273)--	(0.55,0.371)--	(0.6,0.502)--	(0.65,0.681)--	(0.7,0.931)--	(0.75,1.297)--	(0.8,1.866)--	(0.85,2.842)--	(0.9,4.842)--	(0.95,10.954);

 	\draw[smooth,purple,loosely dotted,thick] (0.2,0.029)--	(0.25,0.049)--	(0.3,0.076)--	(0.35,0.112)--	(0.4,0.161)--	(0.45,0.224)--	(0.5,0.308)--	(0.55,0.418)--	(0.6,0.565)--	(0.65,0.766)--	(0.7,1.047)--	(0.75,1.458)--	(0.8,2.096)--	(0.85,3.192)--	(0.9,5.44)--	(0.95,12.313);
\end{tikzpicture}
\caption{Equilibrium bids}\label{fig:equil_bids_a}
\end{subfigure}
\begin{subfigure}{.48\linewidth}
\centering
\begin{tikzpicture}[xscale=5.8,yscale=0.1]
  \def\xmin{0.2}
  \def\xmax{1.1}
  \def\ymin{0}
  \def\ymax{30}
    \draw[->] (\xmin,\ymin) -- (\xmax,\ymin) node[right] {$\rho$} ;
    \draw[->] (\xmin,\ymin) -- (\xmin,\ymax) node[above] {$\W_i$} ;
    \foreach \x in {0.2,0.4,0.6,0.8,1}
    \node at (\x,\ymin) [below] {\x};
    \foreach \y in {0,5,10,15,20,25}
    \node at (\xmin,\y) [left] {\y};

    \draw[smooth,green,thick] (0.2,0.283)--	(0.25,0.387)--	(0.3,0.511)--	(0.35,0.659)--	(0.4,0.835)--	(0.45,1.047)--	(0.5,1.306)--	(0.55,1.626)--	(0.6,2.031)--	(0.65,2.556)--	(0.7,3.259)--	(0.75,4.249)--	(0.8,5.741)--	(0.85,8.236)--	(0.9,13.242)--	(0.95,28.297);

    \draw[smooth,blue,densely dotted,thick] (0.2,0.254)--	(0.25,0.339)--	(0.3,0.437)--	(0.35,0.55)--	(0.4,0.682)--	(0.45,0.839)--	(0.5,1.028)--	(0.55,1.26)--	(0.6,1.55)--	(0.65,1.925)--	(0.7,2.426)--	(0.75,3.13)--	(0.8,4.187)--	(0.85,5.954)--	(0.9,9.496)--	(0.95,20.145);

	\draw[smooth,red,dashed,thick] (0.2,0.241)--	(0.25,0.318)--	(0.3,0.405)--	(0.35,0.505)--	(0.4,0.62)--	(0.45,0.754)--	(0.5,0.915)--	(0.55,1.11)--	(0.6,1.353)--	(0.65,1.664)--	(0.7,2.078)--	(0.75,2.655)--	(0.8,3.52)--	(0.85,4.96)--	(0.9,7.834)--	(0.95,16.447);

	\draw[smooth,black,dashdotted,thick] (0.2,0.233)--	(0.25,0.306)--	(0.3,0.387)--	(0.35,0.478)--	(0.4,0.583)--	(0.45,0.705)--	(0.5,0.848)--	(0.55,1.022)--	(0.6,1.236)--	(0.65,1.51)--	(0.7,1.871)--	(0.75,2.374)--	(0.8,3.123)--	(0.85,4.365)--	(0.9,6.836)--	(0.95,14.219);

 	\draw[smooth,purple,loosely dotted,thick] (0.2,0.23)--	(0.25,0.3)--	(0.3,0.377)--	(0.35,0.464)--	(0.4,0.563)--	(0.45,0.676)--	(0.5,0.81)--	(0.55,0.97)--	(0.6,1.166)--	(0.65,1.415)--	(0.7,1.742)--	(0.75,2.196)--	(0.8,2.868)--	(0.85,3.98)--	(0.9,6.185)--	(0.95,12.756);
\end{tikzpicture}
\caption{Expected waiting times}\label{fig:equil_bids_b}
\end{subfigure}

\begin{subfigure}{0.96\linewidth}
\centering
\begin{tikzpicture}
\begin{customlegend}
    [legend entries={$C_1=0.2$,$C_2=0.4$,$C_3=0.6$,$C_4=0.8$,$C_5=1$},legend columns=-1,legend style={/tikz/every even column/.append style={column sep=0.6cm}}]
    \addlegendimage{green,smooth, thick}
    \addlegendimage{blue,densely dotted,smooth, thick}
    \addlegendimage{red,dashed,smooth, thick}
    \addlegendimage{black,dashdotted,smooth, thick}
    \addlegendimage{purple,loosely dotted,smooth, thick}
    \end{customlegend}
\end{tikzpicture}
\end{subfigure}

\caption{Equilibrium bids and waiting times for increasing values of $\rho$. The other parameters are fixed at $N=5$, $\overline{x}=1$, $\overline{x^2}=2$, and $\frac{\lambda_i}{\lambda}\in(0.2,0.3,0.15,0.25,0.1)$.}
\label{fig:equil_bids}
\end{figure}

\begin{figure}[H]
\centering
\begin{subfigure}{.48\linewidth}
\centering
\begin{tikzpicture}[xscale=5.8,yscale=0.577]
  \def\xmin{0.18}
  \def\xmax{1.1}
  \def\ymin{0.8}
  \def\ymax{6}
    \draw[->] (\xmin,\ymin) -- (\xmax,\ymin) node[right] {$\rho$} ;
    \draw[->] (\xmin,\ymin) -- (\xmin,\ymax) node[above] {$\frac{b_i^e}{b_1^e}$} ;
    \foreach \x in {0.2,0.4,0.6,0.8,1}
    \node at (\x,\ymin) [below] {\x};
    \foreach \y in {1,2,3,4,5,6}
    \node at (\xmin,\y) [left] {\y};

    \draw[green,smooth, thick] (0.2,1)--	(0.25,1)--	(0.3,1)--	(0.35,1)--	(0.4,1)--	(0.45,1)--	(0.5,1)--	(0.55,1)--	(0.6,1)--	(0.65,1)--	(0.7,1)--	(0.75,1)--	(0.8,1)--	(0.85,1)--	(0.9,1)--	(0.95,1)-- (0.99,1);

    \draw[blue,densely dotted,smooth, thick] (0.2,2.639)--	(0.25,2.485)--	(0.3,2.327)--	(0.35,2.207)--	(0.4,2.086)--	(0.45,1.981)--	(0.5,1.889)--	(0.55,1.809)--	(0.6,1.741)--	(0.65,1.681)--	(0.7,1.628)--	(0.75,1.581)--	(0.8,1.539)--	(0.85,1.503)--	(0.9,1.47)--	(0.95,1.441)-- (0.99,1.419);

	\draw[red,dashed,smooth, thick] (0.2,3.806)--	(0.25,3.515)--	(0.3,3.247)--	(0.35,3.029)--	(0.4,2.827)--	(0.45,2.654)--	(0.5,2.505)--	(0.55,2.374)--	(0.6,2.263)--	(0.65,2.166)--	(0.7,2.08)--	(0.75,2.006)--	(0.8,1.94)--	(0.85,1.883)--	(0.9,1.833)--	(0.95,1.789)-- (0.99,1.740);

	\draw[black,dashdotted,smooth, thick] (0.2,4.751)--	(0.25,4.356)--	(0.3,3.998)--	(0.35,3.708)--	(0.4,3.437)--	(0.45,3.208)--	(0.5,3.015)--	(0.55,2.842)--	(0.6,2.697)--	(0.65,2.571)--	(0.7,2.46)--	(0.75,2.363)--	(0.8,2.279)--	(0.85,2.206)--	(0.9,2.144)--	(0.95,2.091)-- (0.99,2.011);

 	\draw[purple,loosely dotted,smooth, thick] (0.2,5.373)--	(0.25,4.922)--	(0.3,4.512)--	(0.35,4.181)--	(0.4,3.875)--	(0.45,3.616)--	(0.5,3.395)--	(0.55,3.2)--	(0.6,3.035)--	(0.65,2.891)--	(0.7,2.766)--	(0.75,2.655)--	(0.8,2.56)--	(0.85,2.478)--	(0.9,2.409)--	(0.95,2.351)-- (0.99,2.248);

\end{tikzpicture}
\caption{Equilibrium bids}\label{fig:equil_bid_ratio_a}
\end{subfigure}
\begin{subfigure}{.48\linewidth}
\centering
\begin{tikzpicture}[xscale=5.8,yscale=3.75]
  \def\xmin{0.18}
  \def\xmax{1.1}
  \def\ymin{0.3}
  \def\ymax{1.1}
    \draw[->] (\xmin,\ymin) -- (\xmax,\ymin) node[right] {$\rho$} ;
    \draw[->] (\xmin,\ymin) -- (\xmin,\ymax) node[above] {$\frac{\W_i}{\W_1}$} ;
    \foreach \x in {0.2,0.4,0.6,0.8,1}
    \node at (\x,\ymin) [below] {\x};
    \foreach \y in {0.4,0.6,0.8,1}
    \node at (\xmin,\y) [left] {\y};

    \draw[green,smooth, thick] (0.2,1)--	(0.25,1)--	(0.3,1)--	(0.35,1)--	(0.4,1)--	(0.45,1)--	(0.5,1)--	(0.55,1)--	(0.6,1)--	(0.65,1)--	(0.7,1)--	(0.75,1)--	(0.8,1)--	(0.85,1)--	(0.9,1)--	(0.95,1)-- (0.99,1);

    \draw[blue,densely dotted,smooth, thick] (0.2,0.898)--	(0.25,0.876)--	(0.3,0.855)--	(0.35,0.835)--	(0.4,0.817)--	(0.45,0.801)--	(0.5,0.787)--	(0.55,0.775)--	(0.6,0.763)--	(0.65,0.753)--	(0.7,0.744)--	(0.75,0.737)--	(0.8,0.729)--	(0.85,0.723)--	(0.9,0.717)--	(0.95,0.712)-- (0.99,0.708);

	\draw[red,dashed,smooth, thick] (0.2,0.852)--	(0.25,0.822)--	(0.3,0.793)--	(0.35,0.766)--	(0.4,0.743)--	(0.45,0.72)--	(0.5,0.701)--	(0.55,0.683)--	(0.6,0.666)--	(0.65,0.651)--	(0.7,0.638)--	(0.75,0.625)--	(0.8,0.613)--	(0.85,0.602)--	(0.9,0.592)--	(0.95,0.581)-- (0.99,0.579);

	\draw[black,dashdotted,smooth, thick] (0.2,0.823)--	(0.25,0.791)--	(0.3,0.757)--	(0.35,0.725)--	(0.4,0.698)--	(0.45,0.673)--	(0.5,0.649)--	(0.55,0.629)--	(0.6,0.609)--	(0.65,0.591)--	(0.7,0.574)--	(0.75,0.559)--	(0.8,0.544)--	(0.85,0.53)--	(0.9,0.521)--	(0.95,0.510)-- (0.99,0.502);

 	\draw[purple,loosely dotted,smooth, thick] (0.2,0.813)--	(0.25,0.775)--	(0.3,0.738)--	(0.35,0.704)--	(0.4,0.674)--	(0.45,0.646)--	(0.5,0.62)--	(0.55,0.597)--	(0.6,0.574)--	(0.65,0.554)--	(0.7,0.535)--	(0.75,0.517)--	(0.8,0.5)--	(0.85,0.483)--	(0.9,0.473)--	(0.95,0.46)-- (0.99,0.45);
\end{tikzpicture}
\caption{Expected waiting times}\label{fig:equil_bid_ratio_b}
\end{subfigure}

\begin{subfigure}{.96\linewidth}
\centering
\begin{tikzpicture}
\begin{customlegend}
    [legend entries={$C_1=0.2$,$C_2=0.4$,$C_3=0.6$,$C_4=0.8$,$C_5=1$},legend columns=-1,legend style={/tikz/every even column/.append style={column sep=0.6cm}}]
    \addlegendimage{green,smooth, thick}
    \addlegendimage{blue,densely dotted,smooth, thick}
    \addlegendimage{red,dashed,smooth, thick}
    \addlegendimage{black,dashdotted,smooth, thick}
    \addlegendimage{purple,loosely dotted,smooth, thick}
    \end{customlegend}
\end{tikzpicture}
\end{subfigure}

\caption{Ratio of equilibrium bids and waiting times of all classes with respect to the lowest cost class ($\frac{b_i^e}{b_1^e}$ and $\frac{\W_i}{\W_1}$), and increasing values of $\rho$. The other parameters are fixed at $N=5$, $\overline{x}=1$, $\overline{x^2}=2$, and $\frac{\lambda_i}{\lambda}\in(0.2,0.3,0.15,0.25,0.1)$.}
\label{fig:equil_bid_ratio}
\end{figure}
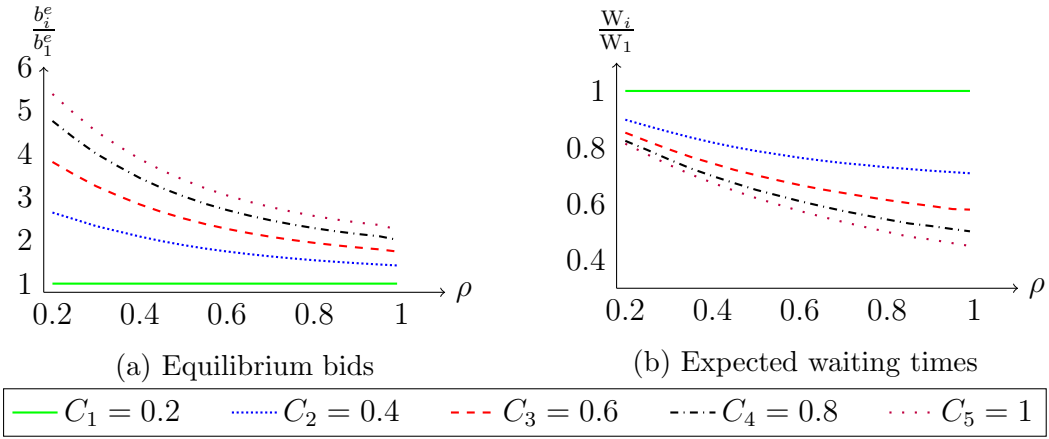

In Figure \ref{fig:equil_bid_ratio}, for the same examples, all equilibrium values are scaled by those of the lowest priority class (i.e. customers with waiting cost $C_1$). The ratio between the bids of the different types decreases with the level of server utilization, while the ratio between waiting times increases. That is, in a busy system customers with lower waiting costs spend more on priority purchasing (than in a less busy system), relative to those with higher waiting costs, but experience relatively longer waiting times. A possible explanation for this is that in a busy system waiting times are typically long enough for a lot of overtaking to take place, even if the difference between accumulation rates is small.

We now wish to dig a bit deeper into the equilibrium properties. It is of interest to study the behaviour of the best response functions determined by the first-order conditions of Theorem \ref{thm:pp_hetero_equil}. We will show a within class analogue of the alternating ATC/FTC phenomena occurs that was discussed in Section \ref{sec:pp_homogeneous}.

\begin{proposition}\label{prop:hetero_tildeW}
Let $b_0=0$ and $b_{N+1}=\overline{b}_N$. For any ordered $\mathbf{b}$ the following properties are satisfied for $i=1,\ldots,N$:
\begin{enumerate}
\item $\tilde{\W}(b_i;\mathbf{b})$ is monotone increasing with respect to $b_i\in(b_{i-1},b_{i+1})$.
\item $\tilde{\W}(b_i;\mathbf{b})$ is monotone decreasing with respect to $b_j\in(b_{j-1},b_{j+1})$, for any $j<i$.
\item $\tilde{\W}(b_i;\mathbf{b})$ is not necessarily monotone in any direction with respect to $b_j\in(b_{j-1},b_{j+1})$, for any $j>i$.
\end{enumerate}
\end{proposition}

We leave this proof for the appendix. Clearly, the waiting $\W(b_i;\mathbf{b})$ decreases with respect to $b_i$ but increases with respect to $b_j$ such that $j\neq i$.  Therefore, if $\tilde{\W}(b_i;\mathbf{b})$ increases (decreases) with $b_j$ then $b_i$ must be increased (decreased) as well to maintain the first-order condition. Proposition \ref{prop:hetero_tildeW} yields the following corollaries.
\begin{corollary}\label{corr:FOC_unique_solution}
If $\mathbf{b}$ is ordered and all $b_j$ such that $j\neq i$ are kept constant then there is at most one solution $b_i\in[b_{i-1},b_{i+1}]$ to the equilibrium condition $\W(b_i;\mathbf{b})=\tilde{\W}(b_i;\mathbf{b})$.
\end{corollary}

Corollary \ref{corr:FOC_unique_solution} implies that a simple bisection can be used in order to numerically compute the symmetric best response of any single class given the strategy of the others. The next question we wish to address is how does this best response behave when the strategies of other types are changed.

Let $\mathbf{b}_{-i}$ denote an ordered profile of all customers excluding type $i$, and denote the solution set, which is either empty or has a single element, by
\[
b_i(\mathbf{b}):=\{b_i:\W(b_i;\mathbf{b}_{-i}\cup b_i)=\tilde{\W}(b_i;\mathbf{b}_{-i}\cup b_i)\}, \quad i=1,\ldots,N.
\]
We can now define a local symmetric best response of type $i$ customers to $\mathbf{b}_{-i}$,
\[
R_i(\mathbf{b}_{-i}):=\left\{
\begin{array}{ll}
b_{i-1}, & \W(b_{i-1};\mathbf{b}_{-i}\cup b_i)>\tilde{\W}(b_{i-1};\mathbf{b}_{-i}\cup b_i), \\
b_i(\mathbf{b}), & b_i(\mathbf{b})\neq\emptyset, \\
b_{i+1}, & \W(b_{i+1};\mathbf{b}_{-i}\cup b_i)>\tilde{\W}(b_{i+1};\mathbf{b}_{-i}\cup b_i).
\end{array}\right.
\]

Combining Proposition \ref{prop:hetero_tildeW} and Corollary \ref{corr:FOC_unique_solution} we get the following result, regarding the local behaviour of the best response functions for ordered bidding profiles.
\begin{corollary}\label{corr:BR_ATC_FTC}
If $\mathbf{b}_{-i}$ are the ordered bids of all customer types excluding $i$, then
\begin{enumerate}
\item For every $i=1,\ldots,N$, $R_i(\mathbf{b}_{-i})$ is monotone increasing (FTC) with $b_j\in[b_{j-1},b_{j+1}]$ such that $j<i$.
\item For every $i=1,\ldots,N$, $R_i(\mathbf{b}_{-i})$ may be increasing or decreasing with $b_j\in[b_{j-1},b_{j+1}]$ such that $j>i$. Thus, it is neither ATC nor FTC.
\end{enumerate}
\end{corollary}

This non-monotone behaviour is illustrated for a numerical example in Figure \ref{fig:pp_BR_hetero}. Specifically, the symmetric best response of type $3$ customers ($R_3$) is computed when a single coordinate is changed from the equilibrium solution, for every one of the other types $i\in\{1,2,4,5\}$.

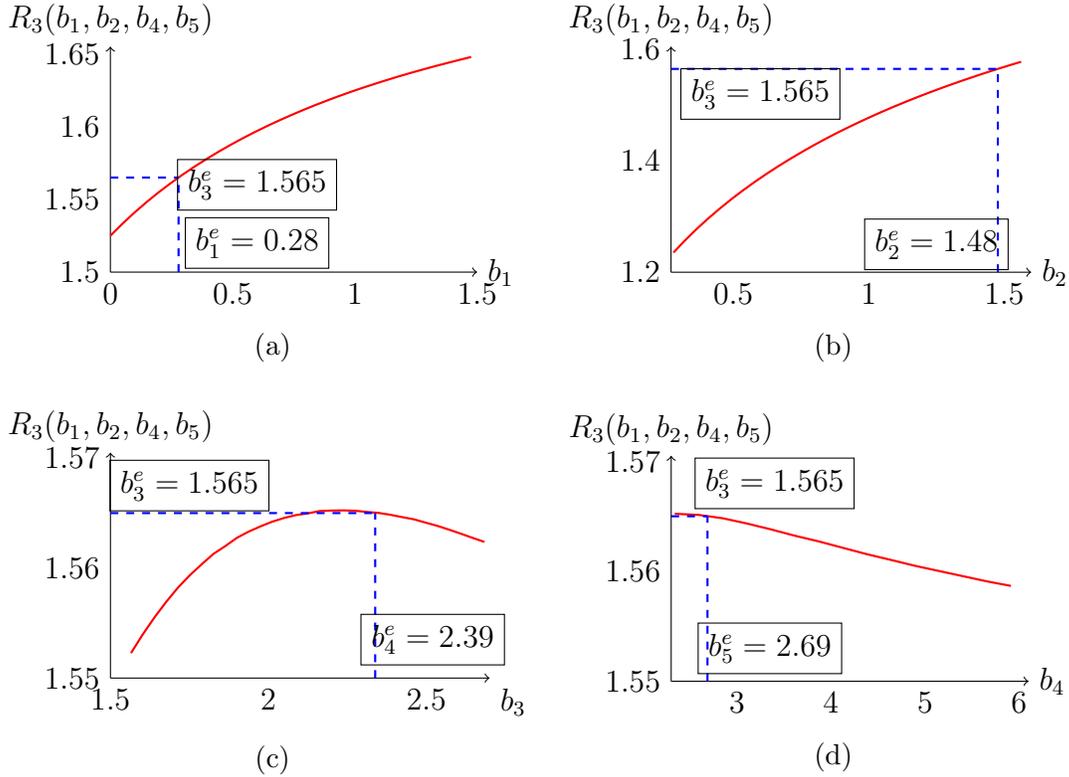
\begin{figure}
\centering
\begin{subfigure}{.48\linewidth}
\begin{tikzpicture}[xscale=3.25,yscale=19.354]
  \def\xmin{0}
  \def\xmax{1.5}
  \def\ymin{1.5}
  \def\ymax{1.655}
    \draw[->] (\xmin,\ymin) -- (\xmax,\ymin) node[right] {$b_1$} ;
    \draw[->] (\xmin,\ymin) -- (\xmin,\ymax) node[above] {$R_3(b_1,b_2,b_4,b_5)$} ;
    \foreach \x in {0,0.5,1,1.5}
    \node at (\x,\ymin) [below] {\x};
    \foreach \y in {\ymin,1.55,1.6,1.65}
    \node at (\xmin,\y) [left] {\y};

        \draw[smooth,red, thick] (0.0001,1.525032)--	(0.0492,1.53323)--	(0.0984,1.540902)--	(0.1476,1.548049)--	(0.1968,1.554776)--	(0.246,1.561029)--	(0.2952,1.566968)--	(0.3444,1.572538)--	(0.3936,1.577793)--	(0.4428,1.582786)--	(0.492,1.587515)--	(0.5412,1.591982)--	(0.5904,1.596239)--	(0.6396,1.600285)--	(0.6888,1.604122)--	(0.738,1.6078)--	(0.7872,1.611321)--	(0.8364,1.614684)--	(0.8856,1.61789)--	(0.9348,1.62099)--	(0.984,1.623933)--	(1.0332,1.626771)--	(1.0824,1.629451)--	(1.1316,1.632079)--	(1.1808,1.634601)--	(1.23,1.637019)--	(1.2792,1.639331)--	(1.3284,1.641591)--	(1.3776,1.643745)--	(1.4268,1.645847)--	(1.476,1.647897);

	\draw[blue, thick, dashed] (0.279,\ymin)--(0.279,1.565);
    \draw[blue, thick, dashed] (\xmin,1.565)--(0.279,1.565);
    \node[draw] at (0.6,1.52) {$b_1^e=0.28$};
    \node[draw] at (0.6,1.56) {$b_3^e=1.565$};

\end{tikzpicture}
\caption{}\label{fig:pp_BR_hetero_1}
\end{subfigure}
\begin{subfigure}{.48\linewidth}
\begin{tikzpicture}[xscale=3.6,yscale=7.41]
  \def\xmin{0.27}
  \def\xmax{1.6}
  \def\ymin{1.2}
  \def\ymax{1.605}
    \draw[->] (\xmin,\ymin) -- (\xmax,\ymin) node[right] {$b_2$} ;
    \draw[->] (\xmin,\ymin) -- (\xmin,\ymax) node[above] {$R_3(b_1,b_2,b_4,b_5)$} ;
    \foreach \x in {0.5,1,1.5}
    \node at (\x,\ymin) [below] {\x};
    \foreach \y in {\ymin,1.4,1.6}
    \node at (\xmin,\y) [left] {\y};

        \draw[smooth,red, thick] (0.2791,1.235287)--	(0.3218,1.257172)--	(0.3646,1.277526)--	(0.4074,1.29652)--	(0.4502,1.314309)--	(0.493,1.331099)--	(0.5358,1.346922)--	(0.5786,1.361922)--	(0.6214,1.376132)--	(0.6642,1.389694)--	(0.707,1.402544)--	(0.7498,1.414878)--	(0.7926,1.426738)--	(0.8354,1.438036)--	(0.8782,1.44888)--	(0.921,1.459207)--	(0.9638,1.469228)--	(1.0066,1.478815)--	(1.0494,1.488109)--	(1.0922,1.497024)--	(1.135,1.505651)--	(1.1778,1.513949)--	(1.2206,1.521979)--	(1.2634,1.529754)--	(1.3062,1.537259)--	(1.349,1.544579)--	(1.3918,1.551598)--	(1.4346,1.558498)--	(1.4774,1.56512)--	(1.5202,1.571612)--	(1.563,1.577853);

    \draw[blue, thick, dashed] (1.477,\ymin)--(1.477,1.565);
    \draw[blue, thick, dashed] (\xmin,1.565)--(1.477,1.565);
    \node[draw] at (1.25,1.25) {$b_2^e=1.48$};
    \node[draw] at (0.6,1.52) {$b_3^e=1.565$};

\end{tikzpicture}
\caption{}\label{fig:pp_BR_hetero_2}
\end{subfigure}

\bigskip
\begin{subfigure}{.48\linewidth}
\begin{tikzpicture}[xscale=4.2,yscale=146.34]
  \def\xmin{1.5}
  \def\xmax{2.7}
  \def\ymin{1.55}
  \def\ymax{1.5705}
    \draw[->] (\xmin,\ymin) -- (\xmax,\ymin) node[below right] {$b_3$} ;
    \draw[->] (\xmin,\ymin) -- (\xmin,\ymax) node[above] {$R_3(b_1,b_2,b_4,b_5)$} ;
    \foreach \x in {1.5,2,2.5}
    \node at (\x,\ymin) [below] {\x};
    \foreach \y in {\ymin,1.56,1.57}
    \node at (\xmin,\y) [left] {\y};

        \draw[smooth,red, thick] (1.5651,1.552281)--	(1.6023,1.553967)--	(1.6396,1.555521)--	(1.6769,1.55694)--	(1.7142,1.558248)--	(1.7515,1.559363)--	(1.7888,1.560355)--	(1.8261,1.561292)--	(1.8634,1.561997)--	(1.9007,1.562754)--	(1.938,1.563325)--	(1.9753,1.563801)--	(2.0126,1.564218)--	(2.0499,1.564557)--	(2.0872,1.564821)--	(2.1245,1.564972)--	(2.1618,1.565191)--	(2.1991,1.565235)--	(2.2364,1.565251)--	(2.2737,1.565209)--	(2.311,1.565114)--	(2.3483,1.565013)--	(2.3856,1.564843)--	(2.4229,1.564639)--	(2.4602,1.564434)--	(2.4975,1.564139)--	(2.5348,1.563837)--	(2.5721,1.56349)--	(2.6094,1.563119)--	(2.6467,1.562743)--	(2.684,1.562383);

    \draw[blue, thick, dashed] (2.338,\ymin)--(2.338,1.565);
    \draw[blue, thick, dashed] (\xmin,1.565)--(2.338,1.565);
    \node[draw] at (2.52,1.5535) {$b_4^e=2.39$};
    \node[draw] at (1.75,1.5675) {$b_3^e=1.565$};

\end{tikzpicture}
\caption{}\label{fig:pp_BR_hetero_3}
\end{subfigure}
\begin{subfigure}{.48\linewidth}
\begin{tikzpicture}[xscale=1.25,yscale=146.34]
  \def\xmin{2.3}
  \def\xmax{6.1}
  \def\ymin{1.55}
  \def\ymax{1.5705}
    \draw[->] (\xmin,\ymin) -- (\xmax,\ymin) node[right] {$b_4$} ;
    \draw[->] (\xmin,\ymin) -- (\xmin,\ymax) node[above] {$R_3(b_1,b_2,b_4,b_5)$} ;
    \foreach \x in {3,4,5,6}
    \node at (\x,\ymin) [below] {\x};
    \foreach \y in {\ymin,1.56,1.57}
    \node at (\xmin,\y) [left] {\y};

        \draw[smooth,red, thick] (2.3381,1.565234)--	(2.5934,1.565128)--	(2.8488,1.564813)--	(3.1042,1.56434)--	(3.3596,1.563815)--	(3.615,1.563237)--	(3.8704,1.562711)--	(4.1258,1.562133)--	(4.3812,1.561555)--	(4.6366,1.561029)--	(4.892,1.560504)--	(5.1474,1.560031)--	(5.4028,1.559558)--	(5.6582,1.559085)--	(5.9136,1.558665);

    \draw[blue, thick, dashed] (2.686,\ymin)--(2.686,1.565);
    \draw[blue, thick, dashed] (\xmin,1.565)--(2.686,1.565);
    \node[draw] at (3.35,1.553) {$b_5^e=2.69$};
    \node[draw] at (3.4,1.568) {$b_3^e=1.565$};

\end{tikzpicture}
\caption{}\label{fig:pp_BR_hetero_4}
\end{subfigure}

\caption{Symmetric best response functions for class $3$ when moving a single coordinate from the equilibrium profile, $\mathbf{b}^e=(0.28,1.48,,1.565,2.39,2.69)$. The other system parameters are $C=(0.2,0.7,0.75,1.25,1.6)$, $\overline{\mathbf{x}}=(0.35,0.85,1,4.5,5)$, $\overline{\mathbf{x}^2}=(2.1,3.7,1.5,21.8,29)$, and $\lambda=(0.06,0.09,0.04,0.07,0.03)$.}\label{fig:pp_BR_hetero}
\end{figure}

This interesting behaviour of the best response functions, together with the non-explicit recursive form of the first-order conditions makes it a cumbersome task to check whether the equilibrium is unique. In particular, verifying that the best response correspondence $R(\mathbf{b}):\mathbbm{R}^N\to \mathbbm{R}^N$ is a contraction mapping. The lack of monotonicity in the first-order conditions implies that there is no submodular form (as defined in \cite{T1979}) which can in some cases be used to show uniqueness or to construct algorithms that converge to equilibrium points (see \cite{Y1995}). We were not able to analytically verify the general diagonally concave conditions of \cite{R1965}, however numerical analysis suggests that they indeed hold for this game (with the proper modification for the non-atomic form). Different numerical methods, such as a naive search on a discretized grid and a best response iteration, converged to a unique equilibrium for all instances tested. The authors believe the equilibrium is unique, but this is left as an open question.

\begin{conjecture}\label{conj_unique}
The Nash equilibrium given by Theorem \ref{thm:pp_hetero_equil} is unique.
\end{conjecture}

%%%%%%%%%%%%%%%%%%%%%%%%%%%%%%%%%%%%%%%%%%%%%%%%%%%%%%%%%%%%%%%
\section{Regulation and pricing}\label{sec:central_opt}
Suppose now that a central planner can set the priority coefficient for each of the customer types, with the goal of minimizing average waiting costs. Thus, we are interested in characterising the profiles
\[
\mathbf{b}^*:=\argmin_{\{\mathbf{b}\in\mathcal{B}\}}\sum_{i=1}^N \lambda_i C_i\W_i(b_i;\mathbf{b}).
\]

If the waiting time sensitivity is linear, then according to the well known $C\mu$ rule (see for example p69 of \cite{book_H2013}), an optimal service regime is one that assigns absolute priority according to the order of the ratios $C_i\frac{1}{\overline{x}_i}$. If $N=2$ this can be achieved by setting $b_1=0$ and $b_2>0$. Otherwise, if $N>2$ an approximation of optimal regime can be obtained by a scaling of rates $\mathbf{b}^{(n)}$, such that\footnote{This approximation was suggested to us by Binyamin Oz.}
\[
\frac{b_{i}^{(n)}}{b_{i-1}^{(n)}}\xrightarrow{n\to\infty} \infty,\quad i=1,\ldots,N,
\]
where $b_0^{(n)}$ is set to equal an arbitrary positive constant. For example, $b_i^{(n)}=\beta^{ni}$ satisfies the condition for any $\beta>1$.

It should be highlighted that the equilibrium can be very far from optimal as the order of the equilibrium bids is determined solely by the order of the waiting time costs, regardless of the expected waiting times (see Lemma \ref{lemma:b_order}). That is, customers with high waiting time costs and long expected service time will potentially purchase a much higher priority than what is socially desired.

It seems that the AP service regime may not be appropriate if minimizing linear waiting costs is the only goal of the central planner. However, it is possible that this is a system constraint, for example for the sake of fairness or if the condition of patients deteriorates while waiting as in healthcare applications. If this is the case we suggest a simple pricing mechanism based on \cite{MW1990}, that will encourage customers to internalise at least some of the disutility caused by their service time. The mechanism prices the AP rate of every customer proportionally to the realised service time. Specifically, a customer purchasing priority $b$ will pay after receiving service of length $x$ a price of $x b$. The cost function of a type $i$ customer, given the profile $\mathbf{b}$ used by all other customers, is then
\[
c_i(b;\mathbf{b})=C_i\W(b;\mathbf{b})+\overline{x}_i b,\quad i=1,\ldots,N.
\]

Clearly, this game is equivalent to the original game with cost $\frac{C_i}{\overline{x}_i}$ for type $i$ customers. The equilibrium is therefore given by Theorem \ref{thm:pp_hetero_equil} with the updated cost. Moreover, from Lemma \ref{lemma:b_order} we have the following property for the above pricing scheme.
\begin{proposition}\label{lprop:SO_pricing_order}
If ${\boldsymbol \pi}$ is an order permutation such that,
\[
\frac{C_{\pi_1}}{\overline{x}_{\pi_1}}<\frac{C_{\pi_2}}{\overline{x}_{\pi_2}}<\cdots<\frac{C_{\pi_N}}{\overline{x}_{\pi_N}},
\]
then the equilibrium profile $\mathbf{b}^e$, under the service time pricing, satisfies
\[
b_{\pi_1}^e<b_{\pi_2}^e<\cdots<b_{\pi_N}^e.
\]
\end{proposition}

Observe that this pricing scheme requires that the customers know their type, but the system administrator does not necessarily have to be able to distinguish between them. If the administrator does know the customer types then simply charging a type $i$ according to $\overline{x}_i$ and not the realised service time would yield the same result. The prices need not necessarily be in the unit of the service time, and can be scaled to any other units by charging a price of $\alpha x b$ for AP rate $b$, where $\alpha>0$ is a constant. More generally, the price can be non-linear: $\alpha(x) b$. This can potentially change the order of the equilibrium bids, for example to take into account the second, or any other, moment of the service distribution.

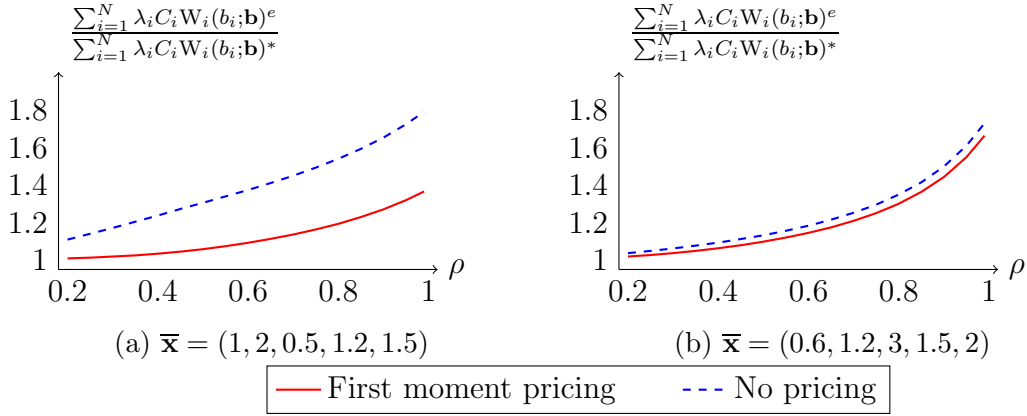
\begin{figure}[H]
\centering
\begin{subfigure}{.48\linewidth}
\begin{tikzpicture}[xscale=6,yscale=2.5]
  \def\xmin{0.18}
  \def\xmax{1.02}
  \def\ymin{0.95}
  \def\ymax{2}
    \draw[->] (\xmin,\ymin) -- (\xmax,\ymin) node[right] {$\rho$} ;
    \draw[->] (\xmin,\ymin) -- (\xmin,\ymax) node[above right] {$\frac{\sum_{i=1}^N\lambda_i C_i\W_i(b_i;\mathbf{b})^e}{\sum_{i=1}^N\lambda_i C_i\W_i(b_i;\mathbf{b})^*}$} ;
    \foreach \x in {0.2,0.4,0.6,0.8,1}
    \node at (\x,\ymin) [below] {\x};
    \foreach \y in {1,1.2,1.4,1.6,1.8}
    \node at (\xmin,\y) [left] {\y};

    \draw[red,smooth, thick] (0.2,1.009)--	(0.25,1.013)--	(0.3,1.019)--	(0.35,1.025)--	(0.4,1.034)--	(0.45,1.045)--	(0.5,1.058)--	(0.55,1.074)--	(0.6,1.092)--	(0.65,1.113)--	(0.7,1.136)--	(0.75,1.163)--	(0.8,1.193)--	(0.85,1.229)--	(0.9,1.27)--	(0.95,1.319)--	(0.99,1.365);

    \draw[blue,dashed, thick] (0.2,1.109)--	(0.25,1.14)--	(0.3,1.171)--	(0.35,1.204)--	(0.4,1.237)--	(0.45,1.272)--	(0.5,1.305)--	(0.55,1.339)--	(0.6,1.373)--	(0.65,1.41)--	(0.7,1.449)--	(0.75,1.492)--	(0.8,1.539)--	(0.85,1.592)--	(0.9,1.652)--	(0.95,1.723)--	(0.99,1.789);

\end{tikzpicture}
\caption{$\overline{\mathbf{x}}=(1,2,0.5,1.2,1.5)$}\label{fig:social_waiting_a}
\end{subfigure}
\begin{subfigure}{.48\linewidth}
\begin{tikzpicture}[xscale=6,yscale=2.5]
  \def\xmin{0.18}
  \def\xmax{1.02}
  \def\ymin{0.95}
  \def\ymax{2}
    \draw[->] (\xmin,\ymin) -- (\xmax,\ymin) node[right] {$\rho$} ;
    \draw[->] (\xmin,\ymin) -- (\xmin,\ymax) node[above right] {$\frac{\sum_{i=1}^N\lambda_i C_i\W_i(b_i;\mathbf{b})^e}{\sum_{i=1}^N\lambda_i C_i\W_i(b_i;\mathbf{b})^*}$} ;
    \foreach \x in {0.2,0.4,0.6,0.8,1}
    \node at (\x,\ymin) [below] {\x};
    \foreach \y in {1,1.2,1.4,1.6,1.8}
    \node at (\xmin,\y) [left] {\y};

    \draw[red,smooth, thick] (0.2,1.019)--	(0.25,1.027)--	(0.3,1.037)--	(0.35,1.049)--	(0.4,1.063)--	(0.45,1.079)--	(0.5,1.098)--	(0.55,1.12)--	(0.6,1.145)--	(0.65,1.174)--	(0.7,1.209)--	(0.75,1.25)--	(0.8,1.3)--	(0.85,1.363)--	(0.9,1.442)--	(0.95,1.548)--	(0.99,1.662);

    \draw[blue,dashed, thick] (0.2,1.037)--	(0.25,1.049)--	(0.3,1.062)--	(0.35,1.077)--	(0.4,1.093)--	(0.45,1.112)--	(0.5,1.132)--	(0.55,1.156)--	(0.6,1.183)--	(0.65,1.215)--	(0.7,1.252)--	(0.75,1.295)--	(0.8,1.349)--	(0.85,1.414)--	(0.9,1.498)--	(0.95,1.609)--	(0.99,1.728);

\end{tikzpicture}
\caption{$\overline{\mathbf{x}}=(0.6,1.2,3,1.5,2)$}\label{fig:social_waiting_b}
\end{subfigure}

\centering
\begin{subfigure}{.5\linewidth}
\begin{tikzpicture}
\begin{customlegend}
    [legend entries={First moment pricing,No pricing},,legend columns=-1,legend style={/tikz/every even column/.append style={column sep=0.8cm}}]
    \addlegendimage{red,thick}
    \addlegendimage{blue,dashed,thick}
\end{customlegend}
\end{tikzpicture}
\end{subfigure}

\caption{Ratio between social welfare in the AP equilibrium and under the absolute $C\mu$ regime, for two different vectors of expected service times. The other parameters are fixed at $N=5$, $\overline{\mathbf{x}^2}=(1.36,5.44,1.5,3.75,5)$, $C_i=(0.2,0.4,0.6,0.8,1.0)$, and $\lambda=(0.16,0.25,0.12,0.21,0.08)$.}
\label{fig:social_waiting}
\end{figure}

In Figure \ref{fig:social_waiting} the effect of the pricing mechanism on the expected waiting times in equilibrium is illustrated for two numerical examples. In Figure \ref{fig:social_waiting_a}, the order of the $C\mu$ rule is completely reversed to that of the waiting costs alone, resulting in significantly worse results without the pricing scheme (up to $\sim 1.8$ higher than optimal without pricing, and up to $\sim 1.2$ higher than optimal with pricing). In Figure \ref{fig:social_waiting_b} we see another example in which expected service times are not homogeneous and the $C\mu$ order is not achieved in equilibrium, but in a milder manner. In this case the pricing scheme does not have much impact on the distance from the optimal waiting times. We can conclude that the pricing scheme is most effective when the order of costs is not aligned with the $C\mu$ order and the expected service times are very non-homogeneous.

%%%%%%%%%%%%%%%%%%%%%%%%%%%%%%%%%%%%%%%%%%%%%%%%%%%%%%%%%%%%%%%
\section{Discussion}\label{sec:conclusion}

This paper has presented and analysed a game of purchasing priorities in an accumulating priority queue. We have shown that if waiting costs are linear then the Nash equilibrium is in pure strategies, and provided a general characterisation that enables a computation of the equilibrium bids. Qualitatively, we showed that both ATC and FTC occur at different bidding levels. If we consider a dynamic play starting at a non-equilibrium point then sometimes customers will compete and outbid each other, and sometimes they will do the opposite and bid less in response to higher bids from others. A numerical analysis shows that in busy systems the customers with low waiting costs will wait much longer on average in equilibrium, even though they make similar bids to the ``higher'' priority customers. Finally, if expected service times differ among types and are not aligned with the waiting time costs (in a $C\mu$ sense) then pricing the AP rates according to realised service times can greatly reduce the congestion levels in the system.

A natural extension of our model is allowing the cost functions to be non-linear, which may provide further motivation for applying this type of dynamic priority regime. This analysis can perhaps be carried out (at least numerically) for some non-linear functions using the machinery of \cite{STZ2014} for the general distribution of the AP queue. Note that the $C\mu$ order may no longer be optimal, and a pricing scheme aiming at a generalized $C\mu$ rule, as suggested and shown to be asymptotically optimal in \cite{VM1995}, can be considered.

There are several other interesting game variations that can be studied for the accumulating priority queue. For instance, customers are exogenously assigned priorities and need to decide whether or not to join. The customers may or may not observe the queue state upon arrival. Other options are considering a relative priority (or relative processor sharing) queue, in which the relative priority is accumulated as the waiting time increases.

%%%%%%%%%%%%%%%%%%%%%%%%%%%%%%%%%%%%%%%%%%%%%%%%%%%%%%%%%%%%%%%
\section*{Acknowledgements}
The authors wish to thank Refael Hassin and an anonymous referee for their helpful comments. The authors gratefully acknowledge the financial support of Israel Science Foundation grant no. 1319/11.

%%%%%%%%%%%%%%%%%%%%%%%%%%%%%%%%%%%%%%%%%%%%%%%%%%%%%%%%%%%%%%%
\bibliography{C:/Users/Liron/Dropbox/University/Research/Full_Bibliography/BigBib}

%%%%%%%%%%%%%%%%%%%%%%%%%%%%%%%%%%%%%%%%%%%%%%%%%%%%%%%%%%%%%%%
\section*{Appendix - proofs}

%%%%%%%%%%%%%
\begin{proof}[Proof of Lemma \ref{lemma:W_general}]
\begin{enumerate}
\item The proof of the recursive waiting time formula \eqref{eq:model_waiting_F} follows the same arguments used in p126-131 of \cite{book_K1976} for the discrete case. We outline the proof without going into all the details. If the AP rate of a singled out customer is $a$ then her expected waiting is
\begin{equation}\label{eq:W_aF}
\begin{split}
\W(a;\mathcal{F}) & = \W_0+\E\sum_{i=1}^N\left[\sum_{j=1}^{N_i(a)}X_{ij}+\sum_{j=1}^{M_i(a)}X_{ij}\right] \\
&= \W_0+\sum_{i=1}^N\overline{x}_i(\M_i(a)+\N_i(a)),
\end{split}
\end{equation}
where $\W_0$ is the expected remaining time of the customer in service, $X_{ij}$ is the service time of the $j$'th type $i$ customer, $\M_i(a)$ is the expected number of type $i$ customers who will overtake her, and $\N_i(a)$ is the expected number of type $i$ customers who were present in the queue upon her arrival and will be admitted into service before her. The second equality in \eqref{eq:W_aF} requires a more cautious consideration. The condition for the Wald identity,
\[
\E\sum_{j=1}^{M_i(a)}X_{ij}=\E M_i(a)\E X_i,
\]
to hold can be stated as (see p158 of \cite{book_d2010})
\[
\E X_m\mathbbm{1}_{\{M_i(a)\geq m\}}=\E X_m\P(M_i(a)\geq m), \ \forall m\geq 1.
\]
This is indeed correct because the event $\{M_i(a)\geq m\}$ simply means that the $m$'th arrival of type $k$ has overtaken the tagged customer, and this only depends on the service times prior to $m$ (recall there is no preemption), and is therefore independent of $X_m$. In a similar manner, if we order the customers present in the queue by their arrival times then the event $\{N_i(a)\geq n\}$ does not depend on $X_n$, but only on the service times $X_1,\ldots,X_{n-1}$.

If the customer arrived at time $0$ and waited for $w\geq 0$ time in the queue then any type $i$ customer with rate $b>a$ that arrived at time $t$ such that $b(w-t)>aw$ has overtaken her. In other words, all arrivals in the interval $\left(0,w\left(1-\frac{a}{b}\right)\right)$ overtake the tagged customer. The arrival rate of such customers is $\lambda_i$, hence by the splitting splitting and superposition properties of the Poisson process and iterating on conditional expectation on the waiting time, we obtain
\begin{equation}\label{eq:app_Ma}
\M_i(a)=\lambda_i \W(a;\mathcal{F})\int_a^\infty\left(1-\frac{a}{b}\right)\ dF_i(b),\quad i=1,\ldots,N.
\end{equation}

We are left with computing the number of type $i$ customers in the queue (upon arrival) that the customer bidding $a$ will not overtake. Clearly she will not overtake any customers who bid $b\geq a$, and the expected number of such customers is $\lambda_i\int_a^\infty \W(b;\mathcal{F})\ dF_i(b)$. A customer with priority $b<a$ who arrived at $t-s$ and waits $v$ time in the queue will still be in the queue if $s<v$ and will not be overtaken if $(v-s)a<vb$. Thus, the probability of a customer with rate $b$ who arrived at $t-s$ being overtaken by the tagged customer who arrived at $t$ is then
\[
\P\left(s<\W(b;\mathcal{F})<\frac{a}{a-b}s\right)=\P\left(\W(b;\mathcal{F})>s\right)-\P\left(\W(b;\mathcal{F})>\frac{a}{a-b}s\right).
\]
Again, we use the properties of the Poisson process, namely splitting and superposition, to obtain the expected number of such customers,
\[
\begin{array}{l}
\lambda_i\int_0^\infty\left[ \P\left(\W(b;\mathcal{F})>s\right)-\P\left(\W(b;\mathcal{F})>\frac{a}{a-b}s\right)\right] ds\ dF_i(b) \\
=\lambda_i\left[\W(b;\mathcal{F})-\left(1-\frac{b}{a}\right)\W(b;\mathcal{F})\right]\ dF_i(b)=\lambda_i\W(b;\mathcal{F})\frac{b}{a}\ dF_i(b)
\end{array}
.
\]
For detailed analysis and justification of the above computations the reader is referred to \cite{book_K1976} and \cite{STZ2014}. By integrating on all customer types we get
\begin{equation}\label{eq:app_Na}
\N_i(a)=\lambda_i\left[\int_0^{a-} \W(b;\mathcal{F})\frac{b}{a}\ dF_i(b)+\int_a^\infty \W(b;\mathcal{F})\ dF_i(b)\right],
\end{equation}
where $a-$ indicates that the integral does include the atom, $dF_i(a)$, if it exists. Combining \eqref{eq:app_Ma} and \eqref{eq:app_Na} we get
\[
\W(a;\mathcal{F})=\W_0+\sum_{i=1}^N\rho_i\left[\int_0^{a-} \W(b;\mathcal{F})\frac{b}{a}\ dF_i(b)+\int_a^\infty\left( \W(a;\mathcal{F})\left(1-\frac{a}{b}\right)+\W(b;\mathcal{F})\ \right)dF_i(b)\right],
\]
or equivalently
\[
\W(a;\mathcal{F})=\frac{\W_0+\sum_{i=1}^N\rho_i\left[\int_0^{a-} \W(b;\mathcal{F})\frac{b}{a}\ dF_i(b)+\int_a^\infty \W(b;\mathcal{F})\ dF_i(b)\right]}{1-\sum_{i=1}^N\rho_i\int_a^\infty\left(1-\frac{a}{b}\right)\ dF_i(b)}.
\]

By the work conservation property we have that
\[
\sum_{i=1}^N\rho_i\int_a^\infty \W(b;\mathcal{F})\ dF_i(b)=\frac{\W_0\rho}{1-\rho}-\sum_{i=1}^N\rho_i\int_0^{a-} \W(b;\mathcal{F})\ dF_i(b),
\]
which leads to the general recursive formula \eqref{eq:model_waiting_F}. Note that $a-$ can be replaced by $a$ because if there is a point mass at $a$ the value inside the integral is zero. Furthermore, since the waiting time is a decreasing function of the bid and $\rho<1$, all above integrals are finite. Specifically, for every $i=1,\ldots,N$ we have
\[
\int_a^\infty \W(b;\mathcal{F})\ dF_i(b) < \W(a;\mathcal{F})\int_a^\infty \ dF_i(b) < \W(a;\mathcal{F})<\infty,
\]
and
\[
\int_a^\infty \left(1-\frac{a}{b}\right)\ dF_i(b) < \int_a^\infty \ dF_i(b)<1.
\]

\item We can rewrite \eqref{eq:model_waiting_F} as
\[
\W(a,\mathcal{F})=\frac{\frac{\W_0}{1-\rho}-\sum_{i=1}^N\rho_i H_i(a)}{1-\sum_{i=1}^N\rho_i J_i(a)},
\]
where
\[
H_i(a):=\int_0^{a}\W(b;\mathcal{F})\left(1-\frac{b}{a}\right)dF_i(b)
\]
and
\[
J_i(a):=\int_a^{\infty}\left(1-\frac{a}{b}\right)dF_i(b).
\]
The only possible points of discontinuity of $H_i(a)$ and $J_i(a)$ are ones such that there is a jump in the measure of integration, i.e. a point $\tilde{a}$ such that $F_i(\tilde{a}-)<F_i(\tilde{a})$ (and $dF_i(\tilde{a})>0$). Observe however, that in these points the value inside the integral is zero. Thus, we have established the continuity of $\W(a,\mathcal{F})$.

The derivative is continuous at $\tilde{a}$ such that $F_i(\tilde{a}-)=F_i(\tilde{a})$ for some $i=1,\ldots,N$ if
\[
\lim_{a\uparrow\tilde{a}}\frac{d}{da}\W(a,\mathcal{F})=\lim_{a\downarrow\tilde{a}}\frac{d}{da}\W(a,\mathcal{F}).
\]

Let $H(a):=\sum_{i=1}^N \rho_i H_i(a)$ and $J(a):=\sum_{i=1}^N \rho_iJ_i(a)$. The first derivative at any point $a$ such that $F_i(a-)=F_i(a), \ \forall i=1,\ldots,N$ is
\[
\begin{split}
\frac{d}{da}\W(a,\mathcal{F}) &= \frac{-H'(a)(1-J(a))+J'(a)\left(\frac{\W_0}{1-\rho}-H(a)\right)}{(1-J(a))^2} \\
&= \frac{\W(a,\mathcal{F})J'(a)-H'(a)}{1-J(a)}.
\end{split}
\]

We have already established that $H(a)$ and $J(a)$ are continuous. Hence, the denominator is continuous. Moreover, the numerator is continuous if
\[
K(a) {:=} \W(a,\mathcal{F})J'(a)-H'(a),
\]
is continuous. With some caution we can apply the derivative chain rule to both integral terms. Using Assumption \ref{assumption:F} that $F$ is defined as a combination of point masses and intervals with positive density we have that
\[
\begin{split}
K(a) &= \sum_{i=1}^N \rho_i\left[ \W(a,\mathcal{F})\frac{d}{da}\int_a^{\infty}\left(1-\frac{a}{b}\right)dF_i(b)-\frac{d}{da}\int_0^{a}\W(b,\mathcal{F})\left(1-\frac{b}{a}\right)dF_i(b)\right] \\
&= -\sum_{i=1}^N \rho_i\left[ \int_a^\infty\W(a,\mathcal{F})\frac{1}{b}dF_i(b)+\int_0^a\frac{b}{a^2}\W(b,\mathcal{F})dF_i(b)\right].
\end{split}
\]

By the right continuity of the \textit{cdf} we have that at any discontinuity point $\tilde{a}$, the term $\W_F(\tilde{a})\frac{1}{\tilde{a}}dF(\tilde{a})$ moves from the left integral to the right integral. Thus, we can conclude that $\lim_{a\uparrow\tilde{a}}K(a)=\lim_{a\downarrow\tilde{a}}K(a)$, for any discontinuity point $\tilde{a}$.
\item First we observe that $K(a)<0$ and therefore $\W(a;\mathcal{F})$ is a monotone decreasing function, as expected. It can further be verified that $J'(a)<0$ and $K'(a)>0$ for all $a>0$, hence
\[
\frac{d^2}{da^2}\W(a;\mathcal{F})=\frac{\rho K'(a)(1-\rho J(a))+K(a)\rho J'(a)}{(1-\rho J(a))^2}>0.
\]

Therefore the expected waiting time of a single customer is strictly convex with respect to a change in her own AP rate, regardless of $F$.
\end{enumerate}
\qed
\end{proof}

%%%%%%%%%%%%%

\begin{proof}[Proof of Proposition \ref{prop:hetero_tildeW}]
We prove the monotonicity properties of
\[
\tilde{\W}(b_i;\mathbf{b})=\frac{1-\sum_{k=i}^N\rho_k\left(1-\frac{b_i}{b_k}\right)-\frac{C_i}{b_i^2}\sum_{k=1}^{i-1}\rho_k b_k \W(b_k;\mathbf{b})}{C_i\sum_{k=i}^N\frac{\rho_k}{b_k}}, \quad i=1,\ldots,N,
\]
where
\[
\W(b_i;\mathbf{b})=\frac{\frac{\W_0}{(1-\rho)}-\sum_{k=1}^{i-1}\rho_k(1-\frac{b_k}{b_i})\W(b_k;\mathbf{b})}{1-\sum_{k=i}^N \rho_k(1-\frac{b_i}{b_k})}, \ 1 \leq i \leq N.
\]
Throughout the proof we assume that $\mathbf{b}$ is ordered and only allow changes of single coordinates that maintain the order, i.e., $b_i\in(b_{i-1},b_{i+1})$ for $i=1,\ldots,N$.
\begin{enumerate}
\item  We will first show that $\tilde{\W}(b_i;\mathbf{b})$ is monotone increasing w.r.t. $b_i$. The result is immediate for $i=1$, as the sum in the numerator is empty. For $i>1$ it suffices to show that
\[
\frac{1}{b_i^2}\W(b_j;\mathbf{b})=\frac{\frac{\W_0}{(1-\rho)}-\sum_{k=1}^{j-1}\rho_k(1-\frac{b_k}{b_j})\W(b_k;\mathbf{b})}{b_i^2\left(1-\sum_{k=j}^N\rho_k\right)+\sum_{k=j}^N\frac{b_j b_i^2}{b_k}},
\]
is monotone decreasing w.r.t. $b_i$ for all $j<i$. The denominator clearly increases with $b_i$. The expected waiting time $\W(b_k;\mathbf{b})$ increases with any $b_j$ such that $j\neq k$, and as this is the only element of the numerator dependent on $b_i$ we have that the numerator is decreasing.
\item Next we will show that $\tilde{\W}(b_i;\mathbf{b})$ decreases with $b_j$ if $j<i$. By taking derivative we have that
\[
\frac{d}{db_j}\tilde{\W}(b_i;\mathbf{b}) \propto -\rho_j\frac{d}{db_j}\left[b_j\W(b_j;\mathbf{b})\right]-\sum_{\{k\neq j,k<i\}}\rho_k b_k\frac{d}{db_j}\W(b_k;\mathbf{b}).
\]
The sum on the right-hand side is positive because $\frac{d}{db_j}\W(b_k;\mathbf{b})$ is positive for all $k\neq j$. It therefore remains to be shown that
\[
b_j\W(b_j;\mathbf{b})=\frac{\frac{\W_0}{1-\rho}-\sum_{k=1}^{j-1}\left(1-\frac{b_k}{b_j}\right)\W(b_k;\mathbf{b})}{\frac{1}{b_j}\left(1-\sum_{k=j+1}^N\rho_k\right)+\sum_{k=j+1}^N\frac{\rho_k}{b_k}},
\]
is an increasing function w.r.t. $b_j$. For $j=1$ this is clearly true, as the sum in the numerator is empty. Furthermore, the the convexity of the waiting time implies that $\W(b_j;\mathbf{b})$ is decreasing at a decreasing rate, and hence cannot decrease faster than $b_j>b_1$. 

\item The last part of the Lemma is a negative result, that can be verified by examples. Numerical examples such as those in Figure \ref{fig:pp_BR} show that the $\tilde{W}_i$ can be both increasing or decreasing with $b_j$ such that $j>i$. Obviously, explicit numbers can be plugged into the equations for this verification.
\end{enumerate}
\qed
\end{proof}

\end{document}